\documentclass[twoside]{article}
\usepackage[a4paper]{geometry}
\usepackage[utf8]{inputenc} 
\usepackage[T1]{fontenc} 
\usepackage[square,sort,comma,numbers]{natbib}
\usepackage{RR}
\usepackage{amsfonts}
\usepackage{amsmath}
\usepackage{amssymb}
\usepackage{amscd}
\usepackage{mathrsfs}
\usepackage{amsthm}
\usepackage{tikz-cd}
\newtheorem{lemma}{Lemma}
\usepackage{graphicx}

\usepackage{algorithm}
\usepackage{algorithmic}
\usepackage{stmaryrd}
\usepackage[symbol]{footmisc}
 \pdfoutput=1
\input xy
\xyoption{all}
\usepackage{tikz}
\usetikzlibrary{positioning}

\def\M{\mathcal{M}}

\def\R{\mathbb{R}}

\def\ones{\mathbf{1}}

\def\zeros{\mathbf{0}}
\usepackage{xspace}

\RRNo{9210}
\RRdate{October 2018}
\RRversion{2}
\RRdater{April 2019}
\RRauthor{
Alain Franc 
  \thanks[fn1]{BIOGECO, INRA, Univ. Bordeaux, 33610 Cestas, France}%
  \thanks[fn2]{Pleiade team - INRIA Bordeaux-Sud-Ouest, France}%
\and Pierre Blanchard
   \thanksref{fn1}
   \thanks[fn3]{HiePACS team, Inria Bordeaux-Sud-Ouest, France}%
\and Olivier Coulaud
    \thanksref{fn3}
}
\authorhead{A. Franc, P. Blanchard P. \& O. Coulaud}
\RRtitle{Géométrie sur les distances et meilleure image euclidienne avec distances pondérées}
\RRetitle{Nonlinear Mapping and Distance Geometry}
\titlehead{Nonlinear Mapping and Distance Geometry}
\RRresume{
Les domaines de géométrie sur les distances (distance geometry) et de recherche de meilleure image euclidienne avec distances pondérées (nonlinear mapping) sont deux domaines classiques : il s'agit pour le premier de construire une isométrie d'un espace métrque discret vers un nuage de points dans un espace euclidien, ne connaissant qu'une partie des distances, et pour le second de construire un nuage avec la meilleure approximation des distances, avec pondération. Nous montrons comment ces méthodes peuvent être rassemblée en une même famille, chacune représentant un choix de pondérations dans un problème d'optimisation. On étudie la continuité entre ces solutions (qui sont des nuages de points), et la compacité des ensembles de solutions (après centrage). On étudie également un exemple numérique, montrant cependant que le problème d'optimisation est loin d'être simple, et que la procédure d'optimisation peut facilement être piégée dans un minimum local. 
}

\RRabstract{
Distance Geometry Problem (DGP) and Nonlinear Mapping (NLM) are two well established questions: Distance Geometry Problem is about finding a Euclidean realization of an incomplete set of distances in a Euclidean space,  whereas Nonlinear Mapping is a weighted Least Square Scaling (LSS) method. We show how all these methods (LSS, NLM, DGP) can be assembled in a common framework, being each identified as an instance of an  optimization problem with a choice of a weight matrix. We study the continuity between the solutions (which are point clouds) when the weight matrix varies, and the compactness of the set of solutions (after centering). We finally study a numerical example, showing that solving the optimization problem is far from being simple and that the numerical solution for a given procedure may be trapped in a local minimum.
}
\RRmotcle{Géométrie sur les distances; Espaces métriques discrets; meilleure image euclidienne avec distances pondérées; optimisation}
\RRkeyword{Distance Geometry; Nonlinear Mapping, Discrete Metric Space, Least Square Scaling; optimization}
\RRprojets{Pleiade \& HiePACS}
\RCBordeaux 

\begin{document}
\makeRR   

\newpage
\tableofcontents

\vspace*{2cm}

\section{Introduction}

Let us have a set $V$ of $n$ objects, and distances between every pair of them. Distance between object $i$ and $j$ is denoted $d(i,j)$, or $d_{ij}$. This defines a metric space $(V,d)$ with $V = \llbracket 1,n \rrbracket$. Among metric spaces, Euclidean spaces play a special role, because they establish a link with geometry. Moreover, the geometry of Euclidean spaces is very well known. Hence, even if many other metric spaces exist, like Riemanian manifolds with distances along a geodesic or graphs with shortest distance between vertices, many efforts have been devoted to specify those metric spaces for which an isometry exists with a Euclidean space, and if it exists, to build it. In such a case, any metric problem in $V$ can be translated into a problem in Euclidean geometry, and, when lucky, solved. For example, supervised learning by discriminant analysis in a discrete metric space can be translated into the same problem in a Euclidean space and solved by Support Vector Machine approaches (see \cite{Cristianini2000}).\\
\\
The conditions for existence of an isometry between a discrete metric space and a subset of a Euclidean space are known, and is given by classical multidimensional scaling, proposed in \cite{Torgerson1952} (see \cite{Mardia1979,Cox2001} for classical presentation of MDS, and \cite{Izenman2008} for a recent presentation). If there is an isometry $i \, : \, i \longmapsto x_i$ between $(V,d)$ and a subset of $n$ points in a Euclidean space $\R^k$, then the Gram matrix of vectors $(x_i)_i$ is definite positive. It appears that the Gram matrix can be computed from pairwise distances only. A set of points $X=(x_i)_i$ such that
\begin{equation}
 \forall \: i,j \in V, \quad \|x_i-x_j\|=d_{ij}
\end{equation}
can be computed from the Singular Value Decomposition of the Gram matrix as a second step. If the Gram matrix has non positive eigenvalues\footnote{In such a case, strictly speaking, the matrix built from the pairwise distances is not a Gram matrix.}, there is no isometry on any subset of $\ell^2$, regardless of  the dimension. In such a case, for a given dimension $k$, one defines the cost of a map $i \longmapsto x_i$ by 
\begin{equation}
 \phi(X) = \sum_{i < j}\left(\|x_i-x_j\|-d_{ij}\right)^2.
\end{equation}

\noindent If there is an isometry between $(V,d)$ and a subset of $\R^k$, there is a map for which $\phi=0$. If not, Least Square Scaling (LSS) is finding a map with minimal cost for a given dimension $k$, i.e. solving\\
\begin{center}
\begin{tabular}{|ll}
 $\quad$ Given & a discrete metric space $(V,d)$ \\
 & a dimension $k$ \\
 $\quad$ find & a map $i \in V \longmapsto x_i \in \R^k$ \\
 $\quad$ such that & $\phi$ is minimal
\end{tabular}
\end{center}
Least square scaling has been pioneered in \cite{Kruskal1964a}. See as well \cite{Cox2001} for a presentation and a comparison with classical MDS\footnote{It is unfortunate that least-square scaling has been proposed with the same name MDS than classical MDS. However, classical texts such as \cite{Mardia1979,Cox2001,Izenman2008} are clear on this matter and agree on setting the vocabulary.}.\\
\\
In some situations, one is interested in the relative erreor/distorsion between the distances $d_{ij}$ and $\|x_i-x_j\|$. Therefore, Sammon has developed in \cite{Sammon69} what he called Non Linear Mapping (NLM) in which each term in the cost function is weighted by the inverse of the distance: 
\begin{equation}
 \phi(X) = \sum_{i < j}\frac{\left(\|x_i-x_j\|-d_{ij}\right)^2}{d_{ij}}
\end{equation}
(Sammon introduces a normalizing constant $c$, that we do not mention here). It is natural to extend this towards
\begin{equation}\label{eq:costfunc}
 \phi(X) = \sum_{i < j}\omega(d_{ij})\left(\|x_i-x_j\|-d_{ij}\right)^2
\end{equation}
where $\omega \; : \; d \longmapsto \omega(d)$ is a weight function, which can be $d^{-1}$ as in Sammon's seminal paper, or other classical maps as $\exp - \beta d$ or $d^{-z}$. Nonlinear mapping is minimizing the cost function over all possible point sets. It is clear that a solution is not unique, and is given up to an isometry in a Euclidean space, i.e. the composition of a translation, a rotation and a reflection.\\
\\
Let us now consider a slightly different situation, where there exists an unknown point set $V$ of $n$ points in a Euclidean space\footnote{Here, it is known as part of the problem that the distances are taken between points living in a Euclidean space, whereas in NLM, such an hypothesis is not required.} $\R^k$ and where distances are known for a subset only of pairs of points. Let us consider the graph $G=(V,E)$ where $E$ is the set of pairs $(i,j) \in V^2$ for which $d(i,j)$ is known. The aim is to find a mapping 
\begin{equation}
\begin{tikzcd}
 x \: : \: V \arrow[r] & \R^k \\
 i \arrow[r, mapsto] & x_i
 \end{tikzcd}
\end{equation}
such that
\begin{equation}\label{eq:dgp}
 \forall \: (i,j) \in E, \quad \|x_i-x_j\|=d_{ij}.
\end{equation}
This is known as Distance Geometry Problem (DGP, see \cite{LLMM14}, equation $(1.1)$). A recent and thorough survey of this problem with historical background on how it grew over decades and different guises is \cite{LLMM14}. See also \cite{Mucherino2013}. Here, the dimension $k$ is given, and the distances are known accurately. In real world problems, such as determination of protein structures, the distances 	are known up to a given precision only. DGP has been studied as an $k-$embeddability problem for graphs. It has been proved in \cite{Saxe1979} that $1-$embeddability problem is NP-complete and $k-$embeddability problem is NP-hard for $k>1$. \\
\\
A link between both problems has been established in \cite{More1997}, where DGP problem is recast as finding a map
\[
\begin{tikzcd}
 x \: : \: V \arrow[r] & \R^k \\
 i \arrow[r, mapsto] & x_i
 \end{tikzcd}
\]
such that
\begin{equation}\label{eq:dgsol}
\phi(X) = \sum_{i,j \in S} \omega_{ij}\left(\|x_i-x_j\|^2-d_{ij}^2\right)^2
\end{equation}
is minimal, where $X$ is the $n \times k$ matrix with $x_i$ in row $i$ and $\omega_{ij}$ are weights. Let us note different choices for the exponents of the quantities to be compared: $\left(\|x_i-x_j\|^2-d_{ij}^2\right)^2$ in equation (\ref{eq:dgsol}) and $\left(\|x_i-x_j\|-d_{ij}\right)^2$ in equation (\ref{eq:costfunc}). Clearly, if DGP has a solution, the minimum of $\phi$ is zero, and a set $(x_1, \ldots,x_n)$ of rows of $X$ where $\phi(X) = 0$  is a solution of the DGP problem. A known difficulty is that $\phi$ has in general many local minima. The technique used in \cite{More1997} is to progressively smooth  $\phi$ by a convolution (see \cite[section 3.2.2]{LLMM14} for a general presentation of smoothing-based methods and of DGSOL which implements it).\\
\\
If all weights are equal to 1 in (\ref{eq:costfunc}), one recovers least-square scaling (see \cite{Cox2001}). If one has 
\[
 \left\{ 
    \begin{array}{lcl}
     (i,j) \in E & \quad \Rightarrow \quad & \omega_{ij} = 1 \\
     (i,j) \notin E & \quad \Rightarrow \quad  & \omega_{ij} = 0
    \end{array}
 \right. 
\]
one recovers DGP. There is a sort of continuity between least-square scaling, nonlinear mapping and DGP when the weights vary smoothly from 1 to 0 on pairs of items outside $E$.\\
\\
Here, we study whether this continuity can be given a sound basis in a relevant topology, and whether it can be translated into a continuity of numerical solutions between NLM and DGP when one or a set of parameters vary.

%
\section{Continuity between LSS, NLM and DGP}\label{sec:continuity}
%

Let $(V,d)$ be a discrete metric space, with $|V|=n$. We denote $d(i,j)=d_{ij}$ with $i,j \in V$. Let $\omega \, : \, d \longmapsto \omega(d) \geq 0$ be a weight function on distances. $(V,d)$ being known, this yields a $n \times n$ weight matrix $\Omega$ of general term $\omega_{ij}$ with $\omega_{ij} = \omega(d_{ij})$ enabling to run a nonlinear mapping of $(V,d)$ in a Euclidean space $\R^k$, where $x_i \in \R^k$ is the image of $i \in V$. The set of points $(x_1,\ldots,x_n)$ is denoted as a $n \times k$ matrix $X$, with $x_i$ being row $i$. The set of real $n \times k$ matrices is denoted $\M(n,k)$. We define
\begin{equation}
  \phi(X,\Omega) = \sum_{i,j=1}^n\omega_{ij}\left(\|x_i-x_j\|^2-d_{ij}^2\right)^2
\end{equation}
and consider the NLM problem 
\begin{equation}\label{pb:nlm:w}
 \left| 
    \begin{tabular}{ll}
     Given & a metric space $(V,d)$ \\
     & a weight matrix $\Omega$ \\
     & a dimension $k$ \\
     find & a point cloud $X \in \M(n,k)$ \\
     such that & $\phi(X,\Omega)$ is minimal
    \end{tabular}
 \right.
\end{equation}
conditioned by $\Omega$. This problem encompasses the problems mentioned above. For example, one can see that
\begin{itemize}
 \item if $\Omega=\ones_n$, i.e. the $n \times n$ matrix with ones only, or $\omega(d)=1 \: \forall d$, problem (\ref{pb:nlm:w}) is least-square scaling (see \cite{Cox2001})
 \item if $\Omega=\zeros_n$, i.e. the $n \times n$ matrix with zeros only, or $\omega(d)=0 \: \forall d$, any point cloud $X$ is a solution
 \item if $\Omega$ is a symmetric boolean matrix, i.e. $\omega_{ij} \in \{0,1\}$, problem (\ref{pb:nlm:w}) is DGP. 
\end{itemize}

\noindent There can be various situations for which considering weighted distances is a useful approach. For example, a formulation with general weights has been recently used in \cite{Mucherino2017}. In some application domains, weights may come from an evaluation of the accuracy of some measurements, as in \cite{Glunt1993} for purpose of NMR applications. Here, our approach is more abstract as we do not consider any privileged application domain, and weights are merely considered as a way to connect by continuity problems known with different names.\\
\\
Let $G=(V,E)$ be the graph such that $(i,j) \in E$ if $\omega_{ij}=1$. Then, in DGP vocabulary, $X$ is a realization of $G$. (see \cite[sec.~1.1.4]{LLMM14}). Let us mention that here and in the following  $\omega(0)$ is undetermined, and one can select $\omega(0)=0$ or $\omega(0)=1$. This can be summarized as follows:
\begin{center}
\begin{tabular}{l|l|l}
$\Omega$ & minimum of $\phi$ $\geq 0$& minimum of $\phi$ $= 0$\\
\hline
$\omega(d_{ij})=1$ & Least square scaling & Isometry with a Euclidean space \\
& & = Classical MDS \\
$\omega(d_{ij}) \geq 0$ & Nonlinear mapping & \\
$\omega(d_{ij})= 1$ if $(i,j) \in E$ & & Distance Geometry Problem \\
\hline
\end{tabular}
\end{center}
This raises the question of the continuity of the solutions of (\ref{pb:nlm:w}) depending on $\Omega$. Definition of continuity is not straightforward as there is a set of matrices $X$ which are solution of (\ref{pb:nlm:w}). We begin by an intuitive notion of continuity, and make it rigorous in section \ref{sec:topology}. A solution is said continuous at $\Omega$ if, $X$ being a solution at $\Omega$, whatever the neighborhood $\mathcal{N}_\textsc{x}$ of $X$, there exists a neighborhood of $\Omega$ such that for each $\Omega'$ in it there is a solution in $\mathcal{N}_\textsc{x}$, or\footnote{Rigorously, one should replace $|X'-X|<\epsilon$ by: there exists an $X'$ in the set of solutions for $\Omega'$ such that $|X'-X|<\epsilon$.}
\[
 \forall \: \epsilon > 0, \quad \exists \: \eta>0 \: : \: |\Omega' - \Omega|< \eta \quad \Longrightarrow \quad |X' - X| < \epsilon.
\]
One can observe that the solution is not continuous at $\Omega=\zeros_n$. Indeed, any point cloud $X'$ is a solution for $\Omega=\zeros_n$. Let us take an $\Omega \neq \zeros_n$ which has a "nice" solution, i.e. the set of solutions is an orbit of the group of isometries acting on $\R^k$, and denote $X$ a solution. Whatever $\eta > 0$, 
\[
 \phi(X, \eta \Omega)= \eta \phi(X, \Omega)
\]
and for any $\eta > 0$, the set of solutions for $\eta\Omega$ still is the set solutions for $\Omega$. Let us take a point cloud $X'$ distant from any $X$ in this set of solution, i.e. $|X'-X|> C$ for some $C>0$ whatever the point cloud $X$ in the set of solutions for $\Omega$. $X'$ is a solution for $\eta = 0$. Whatever the neighborhood of $\zeros_n$, there is a $\eta$ such that $\eta\Omega$ is in this neighborhood. And, for this $\eta\Omega$, there is no point cloud in the neighborhood of $X'$ which is a solution for $\eta\Omega$. Then, the solution of (\ref{pb:nlm:w}) is not continuous at $\Omega=\zeros_n$. Let us note that $X$ and $X'$ are incongruent in the sense of \cite[sect. 1.1.4.3]{LLMM14}, as a congruency class is an orbit of the action of the group of isometries in $\R^k$. \\
\\
Our objective is to study the continuity of the solution of (\ref{pb:nlm:w}), when $\Omega$ varies within the space of matrices with non negative elements. The motivation for this is that each problem has a specific approach to build a solution:
\begin{itemize}
 \item DGP builds it explicitly (often with an optimization scheme)
 \item LSS or NLM uses an optimization scheme
\end{itemize}
We study here whether a continuity between solutions of NLM when a parameter varies in a family $(\Omega_a)_a$ and the limit corresponds to a DGP problem (some distances have zero weight) can lead to a numerical solution of DGP. On the other hand, efficient optimization schemes have been derived for solving DGP which are close to some used for NLM (DC programming, see \cite{LeThi2013} and references 6 \& 7 therein). 

%
\section{A topology on the set of solutions}\label{sec:topology}
%

Let us define
\[
 \psi(\Omega) = \{X \in \M(n,k) \: : \: \phi(X,\Omega) \quad \mbox{is minimal}\}.
\]
For a given $\Omega$, if $m = \min_X \{\phi(X,\Omega)\}$, we have
\[
 \psi(\Omega) = \phi^{-1}(m).
\]
We will define a topology in two spaces, in which: $(i)$ an element is a point cloud $X \in \psi(\Omega)$ and $(ii)$ an element is a set of point clouds. The latter will enable to define the neighborhood of a set $\psi(\Omega)$ for a given $\Omega$ and study continuity of $\psi$. If the topology is associated to a distance, this will read
\begin{equation*}
\forall \: \epsilon >0, \quad \exists \: \eta > 0 \: : \: |\Omega'-\Omega|< \eta \quad \Longrightarrow \quad d(\psi(\Omega'),d(\Omega)) < \epsilon.
\end{equation*}
For this, we must define a distance $d(\psi(\Omega'),d(\Omega))$. A distance can be defined between compact subsets of a metric set: the Hausdorff distance (see e.g. \cite[p. 34]{Krantz2008}). Let us briefly recall its definition. Let $A,B$ be two closed sets in a metric space where the distance between points is denoted $d$. One defines
\[
 \delta(A,B)= \max_{x \in A} \: \left\{\min_{y \in B} \: d(x,y)\right\}.
\]
It is easy to show that $\delta(A,B)=0 \: \Leftrightarrow \: A=B$. But is is not symmetric. So one defines
\[
 d_\textsc{h}(A,B) = \max\: \{\delta(A,B), \delta(B,A\}.
\]
This is defined if $A$ and $B$ are compact subsets, and the triangular identity is fulfilled. $d_\textsc{h}$ is the Hausdorff distance between $A$ and $B$. We first define a distance between two points clouds $X$ and $X'$ as the Hausdorff distance between them. For this to hold, $X$ and $X'$ must be compact. As they are obviously closed, they must be bounded. Therefore, we show a first lemma which gives the condition under which $X \in \psi(\Omega)$ is bounded. Let $G=(V,E)$ be the graph associated to $\Omega$ and defined by
\[
 (i,j) \in E \quad \Longleftrightarrow \quad \omega_{ij} > 0.
\]
Then
\begin{lemma}
The set $X \in \psi(\Omega)$ is bounded if, and only if, $G$ is connected.
\end{lemma}
\begin{proof}
$\bullet\:$ $G$ connected $\Rightarrow$ $X$ is bounded. We show first that $\|x_i-x_j\|$ is bounded for any pair $(x_i,x_j)$ if $(i,j) \in E$. If $m = \min_X \: \phi(X,\Omega)$, we have
\[
 \phi(X,\Omega) = \sum_{i \sim j} \omega_{ij} \left(\|x_i-x_j\|^2-d_{ij}^2\right)^2 = m.
\]
Hence,
\[
 \forall \: i \sim j, \quad \omega_{ij} \left(\|x_i-x_j\|^2-d_{ij}^2\right)^2 \leq m
\]
and
\[
 \|x_i-x_j\|^2 \leq d_{ij}^2 + \frac{m}{\omega_{ij}}.
\]
Let $\delta^2= \max\; d_{ij}^2$ and $\omega = \min \: \omega_{ij}$. Then,
\[
 \forall i \sim j, \quad \|x_i-x_j\|^2 \leq \delta^2 + \frac{m}{\omega}.
\]
Now, let $i,j$ with $i \nsim j$. As $G$ is connected, there is a path $ia_1\ldots a_pj$ linking $i$ with $j$. By triangular inequality, $\|x_i-x_j\| \leq \sum_{k=0}^p \|x_{a_k}-x_{a_{k+1}}\|$ (with $a_0=i$ and $a_{p+1}=j$). As $G$ is finite, it has a diameter $D$, and $\|x_i-x_j\| \leq D\left(\delta^2 + \frac{m}{\omega}\right)$. If there is a $M$ such that for any pair $i,j$ $\|x_i-x_j\| \leq M$ and $\sum_ix_i=0$, then $X$ is bounded.\\
\\
$\bullet\:$ $G$ not connected $\Rightarrow$ $X$ is not bounded. If $G$ is not connected, there are at least two subsets $A,B$ of vertices without connection between $A$ and $B$. We have
\[
 \phi(X,\Omega) = \sum_{i,j \in A} \omega_{ij} \left(\|x_i-x_j\|^2-d_{ij}^2\right)^2 + \sum_{k,m \in B} \omega_{km} \left(\|x_k-x_m\|^2-d_{km}^2\right)^2.
\]
Let us now set 
\[
 x'_i = x_i + \frac{h}{|A|}, \qquad x'_j = x_j - \frac{h}{|B|}.
\]
We still have $\sum_ix'_i + \sum_jx'_j=0$ and, $\forall \: h$, $\phi(X',\Omega)=\phi(X,\Omega)=m$. If $h \rightarrow + \infty$, $X$ is unbounded\\
\\
$\bullet\:$ As a consequence, $X$ is bounded if, and only if, $G$ is connected.
\end{proof}
\begin{flushright}
$\blacksquare$ 
\end{flushright}

\noindent Then, the set of point clouds $X$ solution of (\ref{pb:nlm:w}) for a weight matrix with a connected associated graph is a metric space, with Hausdorff distance. We then consider the subsets $\psi(\Omega)$, running over the matrices $\Omega$ fulfilling connectedness condition. We call $X$ a solution of (\ref{pb:nlm:w}), and $\psi(\Omega)$ the set of solutions. $\psi(\Omega)$ is closed as it is the pre-image of $\{m\}$, the minimum of $\phi$. But it is not bounded. Indeed, $\psi(\Omega)$ is invariant by any isometry in $\R^k$. A translation is an isometry, and the distance between two solutions $X,X'$ such that $X'=h+X$ with $h \in \R^k$ is $\|h\|$. Then, $\psi(\Omega)$ is unbounded. Therefore, we impose to each solutions $X \in \psi(\Omega)$ to be centered, and consider
\[
 \psi_\textsc{c}(\Omega) = \left\{ X \in \psi(\Omega) \: : \: \sum_ix_i = 0 \right\}.
\]
Then $\psi_\textsc{c}(\Omega)$ is bounded. Being a closed and bounded subset of a metric space, it is a compact set. We can use the Hausdorff distance denoted $d_\textsc{h}$ between $\psi_\textsc{c}(\Omega)$ and $\psi_\textsc{c}(\Omega')$, and continuity of the solution of (\ref{pb:nlm:w}) is defined as
\begin{equation}
\forall \: \epsilon >0, \quad \exists \: \eta > 0 \: : \: |\Omega'-\Omega|< \eta \quad \Longrightarrow \quad d_\textsc{h}(\psi(\Omega'),d(\Omega)) < \epsilon.
\end{equation}
It is reasonable to assume that for any pair of points $0 \leq \omega_{ij} \leq 1$. The set of such weight matrices is homeomorphic to the hypercube $[0,1]^N$ with $N=\frac{n(n-1)}{2}$. The weight matrices associated to DGP are boolean. They correspond to vertices of the hypercube.

\section{Continuity and rigidity}

Let us denote by $\mathcal{W}$ the set of weight matrices $\Omega$ such that the associated graph $G=(V,E)$ is connected, and by $\mathcal{X}$ the set of all centered point clouds $X$ with $n$ points. Then, $\psi(\Omega)$ is a compact subset of $\mathcal{X}$.\\
\\
Let $\Omega \in \mathcal{W}$ be the weight matrix of a DGP, i.e. there is a connected graph $G=(V,E)$ such that $\omega(i,j) = 1$ if $(i,j) \in E$ and $\omega(i,j)=0$ if $(i,j) \notin E$. We show here that the continuity of $\psi$ at $\Omega$ is linked with the rigidity of the framework associated to the graph $G$. A framework is a set of $n$ points in $\R^k$ solution of (\ref{eq:dgp}). It is rigid if it is defined up to an isometry (i.e. the set of known distances is sufficient to derive all other distances in a unique way). Otherwise, it is said flexible (see \cite[sect.~1.1.4]{LLMM14} for the definitions, and section 4.2 of same reference for a thorough discussion on those notions). We show here
\begin{lemma}
The set of weight matrices for which any realization in $\R^k$ is a rigid framework is not closed. 
\end{lemma}
\begin{proof}
For this, it suffices to exhibit a sequence $(\Omega_\eta)_\eta$ for which for any $\eta > 0$ DGP solution is a rigid framework which converges to a weight matrix $\Omega=\Omega_0$ for which the solution is flexible. Let us consider $n=3$ and $k=2$ (i.e. 3 points in $\R^2$) and
\[
 D = 
 \begin{pmatrix} 
 0 & 1 & 1 \\ 
 1 & 0 & \sqrt{2} \\ 
 1 & \sqrt{2} & 0 
 \end{pmatrix}
 \qquad \mbox{and} \qquad 
 \Omega_\eta = 
 \begin{pmatrix}
   0 & 1 & 1 \\
   1 & 0 & \eta \\
   1 & \eta & 1
 \end{pmatrix}
 , \qquad 0 \leq \eta \leq 1
\]
where $D$ is the pairwise distance matrix between points. A realization is
\[
 X' = 
 \begin{pmatrix}
  x = & 0 & 0 \\
  y = & 0 & 1 \\
  z = & 1 & 0
 \end{pmatrix}
\]
If $\eta>0$, the realization is defined up to an isometry: the lengths of the edges of the triangle made by $X'$ are known, and this fixes the triangle up to an isometry. Whereas if $\eta=0$,  $\psi(\Omega_0)$ is the set of points $x,y,z \in \R^2$ such that $\|x-y\|=\|x-z\|=1$, which is isomorphic to $\mathbb{O} \times \mathbb{O}$ where $\mathbb{O}$ is the group of rotations in $\R^2$ ($y$ and $z$ are on the circle of center $x$ and radius 1). At $\eta=0$, $\Omega_\eta$ is flexible.
\end{proof}
\begin{flushright}
$\blacksquare$ 
\end{flushright}

As a consequence, $\psi$ is not continuous at $\Omega=\Omega_0$. Indeed, 
\[
 X = 
 \begin{pmatrix}
   0 & 0 \\
   1 & 0 \\
   -1 & 0
 \end{pmatrix}
 \in \psi(\Omega)
\]
and there is a constant $C > 0$ (the calculation is easy, and omitted here) such that
\[
 \forall \: \eta > 0, \quad \forall \: X' \in \psi(\Omega_\eta), \quad d(X,X') > C
\]
whereas $|\Omega-\Omega_\eta|=\eta$. $X$ and $X'$ are incongruent.

%
\section{Convergence to a Heaviside function}
%

We first give succinctly a motivation for next short study of a given family of weight matrices. Molecular based taxonomy consists in assigning specimens to species based on similarities between some relevant DNA sequences, called markers \cite{Hillis96}. Distances between sequences are computed with classical algorithms (see \cite{Gusfield1997}). This dictionary between morphological based and molecular based taxonomy works up to a given threshold of distance (beyond this threshold, computed genetic distances are highly likely to be blurred). Thus, as a simplified setting, distances up to a given threshold only are  taken into account. We have a set of pairwise distances, and we wish to build an Euclidean image of it that is as accurate as possible, ignoring the distances beyond a given threshold. Provided a dimension $k$ is given, this can be formalized as a DGP problem, where in $(V,d)$ distances below a given threshold $\theta$ only are given, i.e. $G$ is defined by
\[
E = \{(i,j) \: : \: d(i,j) \leq \theta\}.
\]
The weights are given by a Heaviside function $w(d)=H(\theta-d)$. Such a function has a discontinuity at $d=\theta$. We can consider the weight function
\begin{equation}\label{eq:weights:tanh}
 \omega_{a,\theta}(d) = \frac{1-\tanh(a(d-\theta))}{2}
\end{equation}
as well because distances are blurred progressively. We have
\[
 \lim_{a \rightarrow \infty} \: \omega_{a,\theta}(d) = H(\theta-d)
\]
with the topology of uniform convergence
\[
 \forall \: \epsilon > 0, \quad \exists \: \alpha \in \R^+ \: : \: \forall \: d > 0, \quad \forall \: a > \alpha, \quad |\omega_{a,\theta}(d) - H(\theta-d)| < \epsilon.
\]
Let us assume that $\theta$ is fixed. We are given a converging family of weight matrices $(\Omega_a)_{a\in \R^+}$ satisfying 
\[
 \Omega_\infty = \lim_{a \rightarrow \infty} \Omega_a.
\]
Then, if $\psi$ is continuous at $\Omega_\infty$, one can define
\[
 \psi(\Omega_\infty) = \lim_{a \rightarrow \infty}\: \psi(\Omega_a).
\]
As $\Omega_\infty$ is the weight matrix of a DGP, i.e. there is a graph $G=(V,E)$ such that $w_{\infty, ij}=1$ for $(i,j) \in E$ and $w_{\infty, ij}=0$ for $(i,j) \notin E$, this means that the solution of DGP can be found as the limit of a family of solutions of NLM, when $a \rightarrow \infty$.\\
\\
Let us make a few remarks on possible frameworks to address such a problem. We call \emph{Euclidean image} a solution of problem (\ref{pb:nlm:w}).
\begin{itemize}
 \item the dimension $k$ is not fixed by the problem itself. Hence, this problem can be rephrased as EMDCP.
 \item It is however interesting to have a Euclidean image in a low dimensional space where shapes of point clouds can be better studied (e.g. $(k \leq 3)$)
 \item In order to build such a solution, we can select a large $k$, and project it in an optimal way in a low dimensional space (by Principal Component Analysis)
 \item The choice to ignore distances beyond a given threshold is similar to the Isomap procedure, which is a manifold learning technique using graph based distances to approximate geodesic distances on a manifold \cite{Tennenbaum2000,Lee2007}. A link between Isomap and solution of a DGP has recently been made in \cite{Liberti2017}.
\end{itemize}
Here, we fix a dimension $k$ (in our example, $k=2$) and study the convergence of NLM solutions to a solution of DGP.  We focus on the numerical stability of such an approach. Indeed, for a given $\Omega$, the manifold $\mathscr{M}$ in $\R^{n \times k} \times \R$ defined by
\[
 \mathscr{M} = \left\{ (X,z) \in \R^{n \times k} \times \R \quad \mbox{s.t.} \quad z = \phi(X,\Omega)\right\}
\]
is far from having a unique minimum on $z$. The set $(\phi^{-1}(m),m)$ when $m$ is the minimum of $\phi$ is a minimal submanifold (all its points have a minimum "elevation" $z$). It is likely that many other local minima exist, as well as "flat valleys". We study here whether the ill behavior of $\mathscr{M}$ may hamper to use the continuity of the solution between NLM and DGP as a way to obtain a numerical solution to DGP. We mention that there exist several efficient methods to solve DGP by numerical optimization, like Difference on Convex functions programming \cite{LeThi2013}, Distance Continuation \cite{More1997} or Isomap-based heuristics \cite{Liberti2017}. We are interested here in the continuity between solutions of NLM and DGP on continuous families of weight matrices $\Omega$.

\section{Numerical optimization schemes}

We have implemented a numerical optimization scheme for solving NLM problem problem with varying weights (defined in equation  \eqref{pb:nlm:w}) with a Heaviside function for the weights (see \eqref{eq:weights:tanh}) for $k=2$, with two methods known to be efficient to find global minima: \texttt{BFGS} and \texttt{basin hopping} \cite{Wales1997} in package \texttt{scipy.optimize}. As distances are unreliable beyond a given threshold, we have adopted the framework of isomap-based heuristics \cite{Liberti2017}. We have built a data set in silico, that we wish to recover (it should be the solution of DGP and NLM problems). The objective is to test the continuity of the numerical solution. The dataset consists of $n=100$ points in $\R^2$ randomly (uniform law) distributed within the ring delimited by circles $r=0.8$ and $r=1$ centered at origin (see FIGURE \ref{fig:shapes} top left). We have selected the 
weight function as defined in equation (\ref{eq:weights:tanh}) with different values of stiffness coefficient $a$ and threshold $\theta$. The numerical results for various values of $a$ and $\theta$ are presented in TABLE \ref{tab:ann} for both BFGS and Basin Hoping. Each cell contains the value of the cost function for the result of the procedure for one value of $a$, one of $\theta$ and one method. The starting point for NLM optimization phase could be the result of Multidimensional Scaling. But this requires the knowledge of all distances. In our case, we wish to avoid the use of distances larger than a given threshold $\theta$ (even if they have been measured). Therefore, the starting point is the point cloud built on distances as computed by Isomap on the partial distance matrix of distances $d \leq \theta$. This yields the following observations:
\begin{itemize}
 \item The cost function for \texttt{Isomap + BFGS} is always equal to or lower than the cost function for \texttt{Isomap + basin hoping}. For this type of problem, \texttt{Isomap + BFGS} is recommended (we have tested other methods, results not shown, like simulated annealing, for which results were worse).
 \item The cost function decreases when $\theta$ increases for a given $a$, and is less sensitive to $\theta$. If $\theta$ is too small, there is a significative probability that the graph of partial distanes is not connected. When $\theta$ and $a$ are low, the optimization step may not converge.
\end{itemize}

\noindent A picture of the datasets obtained by each method (\texttt{Isomap}, \texttt{Isomap + BFGS}, \texttt{Isomap + basin hopping}) is displayed in FIGURE \ref{fig:shapes}.  The eye can recognize a deformed ring in the bottom right graph, namely the best reconstruction with basin hopping. The optimization scheme has been trapped on this pattern. One is tempted to twist the outer small loop to recover a shape close to a ring. The fact that this twist (ouwards like here, or inwards in some other simulations) of a fraction of the ring is a trap can be heuristically understood: the main discrepancy between exact distances and reconstructed distances is for those points which have been twisted outward. They are much closer to the points on the opposite on the ring on the reconstruction than in the initial data set (top left). However, the stiffness of the decrease of the weights for  $d \geq 1$ lets the cost function $\phi$ to be nearly insensitive to those discrepancies. The role of stiffness $a$ is then more important than the role of threshold $\theta$ as the weight of pairs of points separated by a large distance (like $(\|x_i-x_j\| \geq 1$) is annihilated by a very low weight ($\omega_{a,\theta}(1)= 6.7\, 10^{-3}$ for $a=5$ and $\theta=0.75$). Careful observation of many simulations lead to the observation that similar low cost function values may correspond to very different geometric settings for the solution.

\section{Conclusion}

We have set a framework to study continuity of the solution of  NLM (see \eqref{pb:nlm:w}) when the weight matrix $\Omega$ varies. The main point to address is that a solution to NLM is never a single point cloud, but always a set of point clouds, union of orbits of the action of the group of isometries on $\R^k$. We have exhibited an example where the solution is not continuous at a weight matrix where the realization of the solution is not rigid. We expect that it is continuous when realizations are rigid, but this has not been shown here. This is deferred to further work on study of the topological structure of the solutions in relation with the weight matrices. We have linked DGP and NLM in a common framework using this continuity. We have studied whether the continuity of the solution can serve as a basis for ensuring the continuity of numerical solutions when a parameter varies in the weight matrix. Therefore, we have built a simple \emph{in silico} dataset, and derived a procedure with the output of Isomap as initial point for optimization step in NLM, with two optimization schemes (BFGS and Basin Hopping). We have produced good hints to show convergence of NLM solution to DGP solution (a situation when $a \rightarrow \infty$). We have shown as well that different geometric settings of the solution may correspond to similar very low values of the cost function. It is likely that this is due to the complicated shape of the manifold $\mathscr{M}$ of cost function as function of coordinates of a point cloud, and motivates further studies.

\begin{figure}[ht]
\center
\includegraphics[scale=.5]{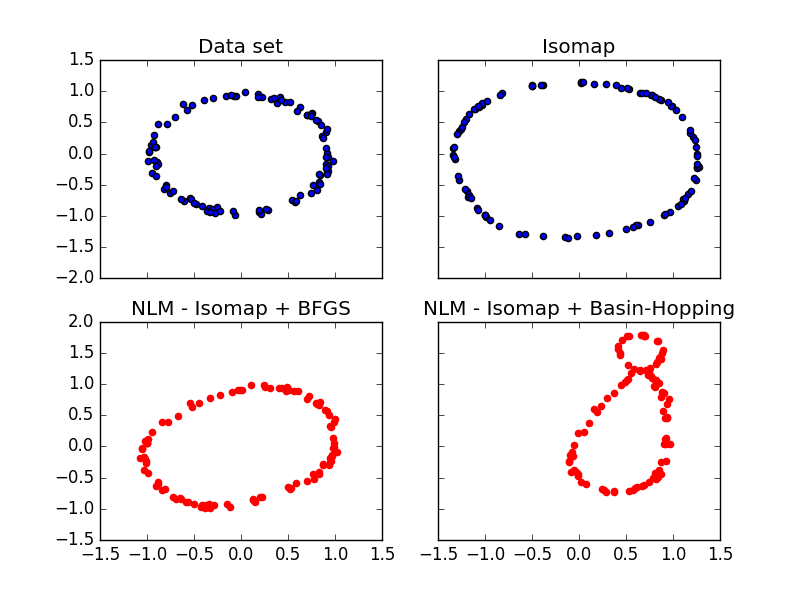}
\caption{Results of numerical simulations for reconstructing the point cloud knowing pairwise distances lower than a given threshold only. Parameters: $\theta = 0.5$ and $a=10$ (see text). Cost function for BFGS is $\phi= 1.018$ and for Basin Hopping $\phi = 2.211$. Top left: Simulated dataset. Top right: dataset reconstructed with standard Isomap procedure. Bottom: dataset reconstructed with BFGS optimization scheme (left) or Basin Hopping (right) with top right dataset as initial values.}
\label{fig:shapes}
\end{figure}

\begin{table}
\begin{small}
\begin{tabular}{c|rcrc|rcrc|rcrc|rcrc}
\hline
	        & $a=5$  &&        && 	$a=10$   &&        && $a=20$            &&                && $a=50$         && 	   \\
	        & BFGS   &$\:$& BH && BFGS       &$\:$& BH && BFGS              &$\:$& BH         && BFGS           &$\:$& BH     \\ 
\hline	        
$\theta=1$      & 94.82  && 94.82  && 0.848      && 10.82  && $9.45 \, 10^{-3}$ && $9.45 10^{-3}$ && $1.83 10^{-3}$ && $1.83 10^{-3}$ \\
$\theta =0.5$   & 55.42  && 55.42  && 1.018      && 2.212  && $1.58 \, 10^{-4}$ && $1.03 10^{-2}$ && $2.85 10^{-5}$ && $4.19 10 ^{-3}$\\
$\theta = 0.25$ &  	 &&        && 0.672      && 0.672  && $1.9  \, 10^{-4}$ && $8.60 10^{-4}$ && $2.12 10^{-5}$ && $1.24 10^{-4}$ \\
\hline
\end{tabular}
\end{small}
\caption{Result of nonlinear mapping: cost function for the reconstruction of an \emph{in silico} dataset, for various values of parameters $a$ and $\theta$, and two methods for optimization phase in NLM: BFGS (for BFGS method) and BH (for Basin Hoping method). The starting point of optimization is the point cloud built by standard Isomap procedure.}\label{tab:ann} 
\end{table}

%
%


\newpage

\bibliographystyle{alpha}
\newcommand{\etalchar}[1]{$^{#1}$}

\end{document}